\documentclass[journal,twocolumn,a4]{IEEEtran}
\usepackage{blindtext, graphicx}
\usepackage{amsmath}
\usepackage{amsthm}
\usepackage{amssymb}

\usepackage[justification=centering]{caption}
\usepackage{algorithm}
\usepackage{arydshln}
\usepackage[noend]{algpseudocode}
\usepackage{multirow}
\usepackage{tikz}
\usepackage[colorinlistoftodos]{todonotes}
\usetikzlibrary{shapes,arrows}
\newtheorem{lemma}{Lemma}

\newtheorem{definition}{Definition}

\begin{document}
\title{Decoding of NB-LDPC codes over Subfields}

\author{
V. B. Wijekoon, Emanuele Viterbo, Yi Hong\\
Monash University, Australia
}

\maketitle

\begin{abstract}
The non-binary low-density parity-check (NB-LDPC) codes can offer promising performance advantages but suffer from high decoding complexity. To tackle this challenge, in this paper, we consider NB-LDPC codes over finite fields as codes over  \textit{subfields} as a means of reducing decoding complexity. In particular, our approach is based on a novel method of expanding a non-binary Tanner graph over a finite field into a graph over a subfield. This approach offers several decoding strategies for a single NB-LDPC code, with varying levels of performance-complexity trade-offs. Simulation results demonstrate that in a majority of cases, performance loss is minimal when compared with the complexity gains.

\end{abstract}



\section{Introduction}
Low-density parity-check (LDPC) codes, which were first introduced by Gallager in 1962 \cite{ldpc:gallager}, have become the error-correcting codes of choice for many practical applications, such as Ethernet, Wi-Fi, and digital television, due to their capacity approaching performance and low-complexity decoding algorithms \cite{ldpc:mckay}. Davey and Mckay introduced the non-binary (NB) counterparts of these codes in 1998 \cite{nb-ldpc:davey-mckay}, and it was soon realized that NB-LDPC codes outperform the binary LDPC codes of comparable length, especially for short-to-moderate code lengths. But these performance gains are yet to be realized in practice due to the high complexity of decoding algorithms.

Best performing decoding algorithm for NB-LDPC codes is the Q-ary sum-product algorithm (QSPA), a generalization of the sum-product algorithm used with binary LDPC codes \cite{nb-ldpc:davey-mckay}. Complexity of QSPA is of the order $\mathcal{O}(q^2)$, where $q$ is the cardinality of the algebraic structure over which the code is defined. However, 
this complexity is too high for most practical applications, and in addition, QSPA also requires a lot of resources, particularly since messages used in decoding are vectors of length $q$. Fast Fourier transform based implementation of QSPA (FFT-QSPA) reduces decoding complexity to $\mathcal{O}(q\log q)$, but still requires similar levels of hardware resources \cite{qspa-fft:declerq}. Log-domain implementations of QSPA (LLR-QSPA) have also been considered in the literature \cite{qspa-llr:wymeersch}.

In \cite{qspa-llr:wymeersch}, the authors introduced a simplified version of LLR-QSPA, referred to as `max-log-SPA', by extending the simplification used in min-sum decoding to NB-LDPC codes and QSPA. This approach was further developed in \cite{ems:declerq-fossorier} by introducing the `Extended Min-Sum' (EMS) algorithm. Instead of considering the complete length $q$ vectors at check node operations, EMS uses the $n_m$ most significant values of each vector, resulting in a complexity order of $\mathcal{O}(n_mq)$. A different approach to simplifying operations of QSPA was proposed in \cite{minmax:savin}, where the `Min-max' algorithm was introduced that has the same complexity order as QSPA, but requires only additions and comparisons. Efficient hardware implementations were proposed for both EMS and min-max algorithms in \cite{trellis-ems:li}, \cite{trellis-minmax:lacruz}.

Expanding the parity-check matrix (PCM) of a NB-LDPC code into a \textit{binary} one allows devising low-complexity bit-level decoding strategies for the non-binary code. Such expansion was proposed in \cite{bec:savin}, called the `extended binary representation'. Additionally, a decoding algorithm for NB-LDPC codes over the binary erasure channel was introduced in \cite{bec:savin}. This strategy is adapted to general channels, as discussed in \cite{generalized-rep:Yu}. In \cite{bit-level:zhang}, the authors used the \textit{binary image} of the non-binary PCM to decode NB-LDPC codes.


While it is possible to construct NB-LDPC codes over many algebraic structures, they are often defined over finite fields, particularly those of characteristic 2 \cite{nb-ldpc:davey-mckay}, i.e. $\mathbb{F}_{2^r}$. A finite field $\mathbb{F}_{p^r}$ contains a unique subfield $\mathbb{F}_{p^m}$, for every $m~|~r$ \cite{book:rothman}. In this paper, we consider expanding the PCM of a NB-LDPC code over $\mathbb{F}_{p^r}$ to a matrix over any such subfield $\mathbb{F}_{p^m}$. For codes over $\mathbb{F}_{2^r}$, this includes the expansion into $\mathbb{F}_{2}$, a binary expansion. We then propose a general decoding algorithm, usable with any of the many possible expansions. Since the operations of the decoder would now be with a smaller size field, significant gains in complexity is achievable, and simulation results demonstrate that the performance loss in comparison to QSPA is minimal. Also, it now becomes possible to decode the same code over several different fields, each offering a different performance-complexity trade-off. 

The remaining of the paper is organized as follows. Section II introduces the mathematical concepts used for the expansion, while Section III provides the expansion along with examples. Section IV presents the decoding strategy, whereas Section V includes simulation results. Section VI analyzes the complexity and resource requirements of the new decoding scheme, and Section VII concludes the paper.

\section{$\alpha$-connected Subgroups}
Consider a finite field of characteristic $p$, $\mathbb{F}_{p^r}$, and one of its additive subgroups, $G$. We denote a primitive element of $\mathbb{F}_{p^r}$ with $\alpha$. It is easy to verify that multiplying all the elements in $G$ with some $\alpha^i \in \mathbb{F}_{p^r}$ yields another additive subgroup. Then we have the following definition.

\begin{definition}
Additive subgroups $G_1$ and $G_2$ of $\mathbb{F}_{p^r}$, where one can be obtained from the other by multiplying by some power of $\alpha$, are called $\alpha-${\bf\em connected subgroups}.
\end{definition}

If subgroups $G_1$ and $G_2$ are $\alpha$-connected, and so are $G_2$ and $G_3$, then clearly $G_1$ and $G_3$ are also $\alpha$-connected. This yields the following denifition. 

\begin{definition}
A set of subgroups $S=\{G_1,..,G_n\}$ of $\mathbb{F}_{p^r}$ is an $\alpha$-connected set if
\begin{enumerate}
    \item {
        each $G_i\in S$ is $\alpha$-connected with all other $G_j\in S$.
    }
    \item {
        $G_i \in S$ is $\alpha$-connected with some $G$, then $G\in S$.
    }
\end{enumerate}
\end{definition}

An $\alpha$-connected set $S$ can be generated using any $G_i\in S$, simply by multiplying with increasing powers of $\alpha$. Each generated subgroup will be added to the set until for some power $i^*$, $\alpha^{i^*}.G_i$ results in $G_i$ itself. This $i^*$ would be the cardinality of the set $S$, which we denote with $|S|$. Lemma 1 considers the minimum possible cardinality of an $\alpha$-connected set.

\begin{lemma}
Consider $\mathbb{F}_{p^r}$ and let $m~|~r$. Then the smallest possible $\alpha$-connected set of additive subgroups of order $p^{r-m}$ has a cardinality of $\frac{p^r-1}{p^m-1}$.  
\end{lemma}
\begin{proof}
Let $G$ be an additive subgroup of $\mathbb{F}_{p^r}$, and $|G|=p^{r-m}$. Assume $\alpha^i\cdot G=G$ for some values $i\in \{1,...,p^r-2\}$. Minimum possible cardinality of an $\alpha$-connected set is the minimum possible value $i$ satisfying $\alpha^i\cdot G=G$, denoted by $i_m$.

Let $S_{\alpha^{i_m}}$ be the set of elements in $\mathbb{F}_{p^r}$ generated by $\alpha^{i_m}$. Note that $i_m$ is the minimum non-zero power of $\alpha$  in $S_{\alpha^{i_m}}$. If $\alpha^{i_m}\cdot G=G$, then clearly $\alpha^{ki_m}\cdot G=G$, for any $\alpha^{ki_m}\in S_{\alpha^{i_m}}$. Since we are focused on the minimum, we only consider $i_m$, for which the following relation holds.
\[ i_m|S_{\alpha^{i_m}}|=(p^r-1) \tag{1}\]

Since we assume $\alpha^{i_m}\cdot G=G$, if $g\in G$, then $g\cdot S_{\alpha^{i_m}} \subset G$. As $G$ is an additive subgroup, it must contain the additive identity $0$, and $0\cdot S_{\alpha^{i_m}}=\{0\}$. Note that for $g_1,g_2\in G$ that are both $\neq0$, sets $g_1.S_{\alpha^{i_m}}$ and $g_2.S_{\alpha^{i_m}}$ would be of the same size, and they should either be the same set or disjoint. Then, as order of $G$ is $p^{r-m}$, disregarding $0$, following should hold for some value $n$.
\[ n|S_{\alpha^{i_m}}| = (p^{r-m}-1) \tag{2} \]

From (1) and (2), we see that $|S_{\alpha^{i_m}}|$ is a factor of both $(p^r-1)$ and $(p^{r-m}-1)$. Since (1) shows that $i_m$ and $|S_{\alpha^{i_m}}|$ are inversely proportional, $|S_{\alpha^{i_m}}|$ should be the \textit{greatest common divisor} of $(p^r-1)$ and $(p^{r-m}-1)$. We note the following.
\begin{align}
(p^r-1)=(p^m-1)\sum_{i=0}^{\frac{r}{m}-1}(p^m)^i\nonumber \\
(p^{r-m}-1)=(p^m-1)\sum_{i=0}^{\frac{r}{m}-2}(p^m)^i\nonumber\\
\sum_{i=0}^{\frac{r}{m}-1}(p^m)^i=p^m\sum_{i=0}^{\frac{r}{m}-2}(p^m)^i+1\nonumber
\end{align}

Above allows us to conclude that $\gcd(p^r-1,p^{r-m}-1)=(p^m-1)$, and using (1);
\[ i_m=\frac{p^r-1}{p^m-1} \]
\end{proof}

Even though lemma 1 shows that the smallest $\alpha$-connected set should be of cardinality $\frac{p^r-1}{p^m-1}$, it does not reveal how to construct such a set. It should also be noted that the additive property of $G$ is not necessary for the proof. Only the existence of the additive identity is used.

Lemma 2 outlines a method to construct an $\alpha$-connected set of subgroups of order $p^{r-m}$.

\begin{lemma}
Let $G$ be a subgroup of order $p^m$ in $H'=\{\mathbb{F}_{p^r}, +\}$, and $\psi$ some surjective homomorphism $\psi~:~H'\rightarrow G$. The kernels of the set of homomorphisms  $\psi_i(h')=\psi(\alpha^{-i}h')$, for $i=\{0,1,...,p^r-2\}$, form an $\alpha$-connected set of additive subgroups of order $p^{r-m}$. 
\end{lemma}
\begin{proof}
For $\psi~:~H'\rightarrow G$, $\ker(\psi)$ is an additive subgroup of order $p^{r-m}$ in $\mathbb{F}_{p^r}$ \cite{book:rothman}. Since $\psi_i(h')=\psi(\alpha^{-i} h')$, it is clear that $\ker(\psi_i)=\alpha^i \ker(\psi)$. Thus, $\ker(\psi)$ and $\ker(\psi_i)$ are $\alpha$-connected subgroups for all possible $i$. Also, for any $i_1,i_2$, $\ker(\psi_{i_1})=\alpha^{i_1-i_2}\ker(\psi_{i_2})$. Then, the set of kernels $S=\{ \ker(\psi_0),\ker(\psi_1),...,\ker(\psi_{p^r-2})\}$, where possible duplicates have been removed, satisfy the conditions of Definition 2, and thus form a $\alpha$-connected set.
\end{proof}
\vspace{2mm}

Cardinality of an $\alpha$-connected set generated as in Lemma 2 depends on the homomorphism $\psi$. Therefore, to construct the smallest $\alpha$-connected set, one must find a suitable homomorphism. The homomorphism we use is based on the representation of $\mathbb{F}_{p^r}$ as an \textit{extension} of the subfield $\mathbb{F}_{p^m}$. The following Lemma establishes the structure of $\mathbb{F}_{p^m}$ in $\mathbb{F}_{p^r}$.

\begin{lemma}
Let $S_\beta$ be the set of elements in $\mathbb{F}_{p^r}$ generated by $\beta = \alpha^{\frac{p^r-1}{p^m-1}}$. Then, $S_\beta \cup \{0\}$, where $0$ is the additive identity of $\mathbb{F}_{p^r}$, is the subfield $\mathbb{F}_{p^m}$.
\end{lemma}
\begin{proof}
Since $m~|~r$, $\mathbb{F}_{p^r}$ should contain the subfield $\mathbb{F}_{p^m}$. Let $K'=\{\mathbb{F}_{p^r}, \times\}$ and $K=\{\mathbb{F}_{p^m}, \times\}$. $K$ is a subgroup of order $(p^m-1)$ of $K'$. Note that both $K$ and $K'$ are cyclic. From properties of subgroups of cyclic groups \cite{book:rothman}, there should be only one unique subgroup of a specific order in $K'$. Set of elements generated by $\beta$, $S_\beta$, is such a subgroup, of order $(p^m-1)$, and thus $K=S_\beta$. This allows the conclusion $\mathbb{F}_{p^m}=S_\beta \cup \{0\}$.
\end{proof}
\vspace{2mm}

We are interested in the polynomial representation of $\mathbb{F}_{p^r}$ as an extension of $\mathbb{F}_{p^m}$. In such a representation, some $\alpha^i\in \mathbb{F}_{p^r}$ is represented with a polynomial $E_{\alpha^i}(x)$ over $\mathbb{F}_{p^m}$, of degree at most $(\frac{r}{m}-1)$. In the case of elements belonging to $\mathbb{F}_{p^m}$ (for $\beta^i$), the polynomials would be of degree $0$. Primitive polynomial for the representation, $\Pi(x)$, is of degree $\frac{r}{m}$. Also note that since $\Pi(x)$ is irreducible, it should have a non-zero constant term. Based on this representation, we define a homomorphism $\psi^*$ between the additive groups of $\mathbb{F}_{p^r}$ and $\mathbb{F}_{p^m}$ as follows.

\begin{definition}
Let $H'=\{\mathbb{F}_{p^r}, +\}$, and $H=\{\mathbb{F}_{p^m}, +\}$. Homomorphism $\psi^*~:~H'\rightarrow H$ is mapping $h'\in H'$ to $h\in H$ if the constant term in $E_{h'}(x)$ is $h$.
\end{definition}

Using homomorphism $\psi^*$ in the method proposed in Lemma 2 generates an $\alpha$-connected set of minimum cardinality.

\begin{lemma}
The set of kernels of homomorphisms $\psi_i^*(h')=\psi^*(\alpha^{-i}h')$, for $i=\{0,1,...,p^r-2\}$ form an $\alpha$-connected set of additive subgroups of order $p^{r-m}$ of the minimum cardinality $\frac{p^r-1}{p^m-1}$.
\end{lemma}
\begin{proof}
$\psi^*~:~H'\rightarrow H$, where $H'=\{\mathbb{F}_{p^r}, +\}$ and $H=\{\mathbb{F}_{p^m}, +\}$. Since $|H|=p^m$, $\ker(\psi^*)$ is a subgroup of order $p^{r-m}$ in $H'$. From lemma 2, it is clear that kernels of $\psi_i$, where $\psi_i^*(h')=\psi^*(\alpha^{-i}h')$ form an $\alpha$-connected set. Cardinality of this set is equal to the minimum value of $i$ for which $\ker(\psi^*_i)=\alpha^i \ker(\psi^*_0)=\ker(\psi^*_0)$, which we denote with $i_m$.

Let $g_j\in \ker(\psi^*_0)$ and $\alpha^{i_m}g_j=\gamma_j$ ($j=1,...,p^{r-m}$). Let polynomial representations (in the extended representation) of $\alpha^{i_m},g_j$ and $\gamma_j$ be, respectively, $E_{\alpha^{i_m}}(x), E_{g_j}(x)$ and $E_{\gamma_j}(x)$. These are related as follows, where $K_j(x)$ is some polynomial over $\mathbb{F}_{p^m}$.
\[ E_{\alpha^{i_m}}(x)E_{g_j}(x) = \Pi(x)K_j(x) + E_{\gamma_j}(x)\]
Since $g_j\in \ker(\psi^*_0)$, constant term in $E_{g_j}(x)$ is zero, which makes the constant term in $E_{\alpha^{i_m}}(x)E_{g_j}(x)$ also zero. Note that for $\gamma_j \in \ker(\psi^*_0)$, constant term of $E_{\gamma_j}(x)$ should be zero. As observed earlier, $\Pi(x)$ has a non-zero constant term, and therefore, for $\gamma_j \in \ker(\psi)$, $K_j(x)$ should be a polynomial with a zero constant term. For $\alpha^{i_m} \ker(\psi^*_0)=\ker(\psi^*_0)$, this should be true for $j=1,..,p^{r-m}$.

Polynomial representations of elements in $\ker(\psi^*_0)$ contains at least one polynomial of each possible degree, from $0$ to $\frac{r}{m}-1$. Then, if $\deg(E_{\alpha^{i_m}}(x))>0$, for at least one value of $j$, $E_{\alpha^{i_m}}(x)E_{g_j}(x)$ would be of degree $\frac{r}{m}$. Since $\Pi(x)$ is also of degree $\frac{r}{m}$, this requires $K_j(x)$ to be a non-zero constant for that particular value of $j$, resulting in $\alpha^{i_m} \ker(\psi^*_0)\neq \ker(\psi^*_0)$. Therefore, for $\gamma_j \in \ker(\psi^*_0)$ for all $j=1,..,p^{r-m}$, $\deg(E_{\alpha^{i_m}}(x))=0$. In such a case, $K_j(x)=0$ for all $j$. This requires $\alpha^{i_m} \in \mathbb{F}_{p^m}$, and since we require the \textit{minimum},  $i_m=\frac{p^r-1}{p^m-1}$.

Thus, using homomorphism $\psi^*$ as in lemma 2, it is possible to construct an $\alpha$-connected set of additive subgroups of order $p^{r-m}$, that has the minimum cardinality $\frac{p^r-1}{p^m-1}$, as proved in Lemma 1.
\end{proof}

\section{Graph Expansion}

In this section, we will present how a graph over $\mathbb{F}_{p^r}$ can be expanded into a larger one over $\mathbb{F}_{p^m}$, where $m~|~r$, using the smallest set of $\alpha$-connected subgroups of order $p^{r-m}$ in $\mathbb{F}_{p^r}$, constructed as detailed in the previous section. We represent this special $\alpha$-connected set by $\Theta_{p^{r-m}}$ from here onwards. Basic mathematical concepts used in the expansion are briefly over-viewed in subsection A, while the expansion is presented in subsection B, along with an example.

\subsection{Preliminaries}
Consider some surjective homomorphism $\psi~:~H'\rightarrow H$, where $H'=\{ \mathbb{F}_{p^r},+ \}$ and $H=\{ \mathbb{F}_{p^m},+ \}$. As remarked earlier as well, $\ker(\psi)$ is a subgroup of $H'$, of order $p^{r-m}$. Since the homomorphism is surjective, according to the first isomorphism theorem \cite{book:rothman}, quotient group $Q_{\psi} = H'/\ker(\psi)$ is isomorphic to $H$. $Q_{\psi}$ contains the $p^m$ cosets of $\ker(\psi)$, including the trivial coset ($\ker(\psi)$ itself). In the isomorphism between $Q_{\psi}$ and $H$, this trivial coset would map to the identity element of $H$ (the additive identity of $\mathbb{F}_{p^m}$), and the other cosets would map to the remaining elements of $H$. 

Let $Q_{\psi} = \{ C^0_{\psi}, C^1_{\psi},..., C^{p^m-1}_{\psi}\}$, where each $C^j_{\psi}$ represents some coset of $\ker(\psi)$, with $C^0_{\psi}$ representing the trivial coset. Cosets contain elements in $\mathbb{F}_{p^r}$, and using the multiplicative properties of the field, we define a `multiplication' operation on $Q_{\psi}$ as follows.

\begin{definition}
Operation $\beta Q_{\psi}$, for some $\beta \in \mathbb{F}_{p^r}$, is defined as $\beta Q_{\psi} = \{ \beta C^0_{\psi},\beta C^1_{\psi}, ..., \beta C^{p^m-1}_{\psi}\}$. 
\end{definition}

Given that two subgroups in $H'$ are $\alpha$-connected, then the respective quotient groups are also related in a similar way, as shown in the following lemma.

\begin{lemma}
Let $G_{\psi_1}$ and $G_{\psi_2}$ be two $\alpha$-connected subgroups of $H'$, with $G_{\psi_1}=\alpha^k G_{\psi_2}$. Also, $Q_{\psi_1}=H'/G_{\psi_1}$ and $Q_{\psi_2}=H'/G_{\psi_2}$ are the two quotient groups. Then, $\alpha^k Q_{\psi_2}$ is a \textit{permutation} of $Q_{\psi_1}$.
\end{lemma}
\begin{proof}
Let $Q_{\psi_1} = \{ C^0_{\psi_1},...,C^{p^m-1}_{\psi_1} \}$ and $Q_{\psi_2} = \{ C^0_{\psi_1},...,C^{p^m-1}_{\psi_1} \}$. Here the trivial cosets $C^0_{\psi_j}$ are the subgroups themselves, and all the cosets can be represented using the respective subgroup and some coset leader term as follows:
\[ Q_{\psi_1} = \{ G_{\psi_1},G_{\psi_1}+l^1_{\psi_1},...,G_{\psi_1}+l^{p^m-1}_{\psi_1} \} \]
\[ Q_{\psi_2} = \{ G_{\psi_2},G_{\psi_2}+l^1_{\psi_2},...,G_{\psi_1}+l^{p^m-1}_{\psi_2} \} \]

Using the multiplication operation on $Q_{\psi_2}$ yields
\[ \alpha^k Q_{\psi_2} = \{ \alpha^k G_{\psi_2},\alpha^k G_{\psi_2}+\alpha^k l^1_{\psi_2},...,\alpha^k G_{\psi_2}+\alpha^k l^{p^m-1}_{\psi_2} \} \]

As cosets of any subgroup are mutually exclusive, and due to the multiplicative properties of the field, any $\alpha^k G_{\psi_2} + \alpha^k l^j_{\psi_2}$ is disjoint with any other. Since $G_{\psi_1}=\alpha^k G_{\psi_2}$ then
\[ \alpha^k Q_{\psi_2} = \{ G_{\psi_1},G_{\psi_1}+\alpha^k l^1_{\psi_2},...,G_{\psi_1}+\alpha^k l^{p^m-1}_{\psi_2} \} \]

$\alpha^k Q_{\psi_2}$ is a set containing $G_{\psi_1}$, and its $(p^m-1)$ proper cosets, albeit the coset leader terms could have changed. Thus, $Q_{\psi_1}$ and $\alpha^k Q_{\psi_2}$ are the \textit{same sets}. Comparing the original representation of $Q_{\psi_1}$ with $\alpha^k Q_{\psi_2}$, it is apparent that the position of the trivial coset is not changed, but there is no such guarantee for other cosets. Thus, when elements of the quotient groups are considered in some specific order, then $\alpha^k Q_{\psi_2}$ is some permutaion of $Q_{\psi_1}$.
\end{proof}

The homomorphisms $\psi^*_i$ we use in constructing the smallest $\alpha$-connected set, $\Theta_{p^{r-m}}$, are all surjective. $\Theta_{p^{r-m}}$ consists of the kernels of these homomorphisms, and it is possible to construct a \textit{set of quotient groups} with those kernels. Let that set be $\Theta^Q_{p^{r-m}} = \{ Q_{\psi^*_i}~;~i=0,...,\frac{p^r-1}{p^m-1}-1\}$. Since all $\psi^*_i$'s are surjective, each $Q_{\psi^*_i}$ is isomorphic to $H$, the additive group of $\mathbb{F}_{p^r}$. Since $\Theta_{p^{r-m}}$ is $\alpha$-connected, according to Lemma 5, multiplying some $Q_{\psi^*_i}\in \Theta^Q_{p^{r-m}}$ by some power of $\alpha$ results in a permutation of some other $Q_{\psi^*_j}\in \Theta^Q_{p^{r-m}}$.


These observations about $\Theta^Q_{p^{r-m}}$ provide some insights on how to decode a code over $\mathbb{F}_{p^r}$ over one of its subfields $\mathbb{F}_{p^m}$. Instead of traditionally used symbol probabilities, we consider the probabilities of a variable node \textit{belonging} to each coset of each quotient group in $\Theta^Q_{p^{r-m}}$. Then, for each variable node, $\frac{p^r-1}{p^m-1}$ probability vectors of length $p^m$ are required, which we refer to as \textit{`coset probability vectors (CPVs)'}. Complexity bottleneck in decoding NB-LDPC codes are the check node operations \cite{ems:declerq-fossorier}-\cite{minmax:savin}, and advantages of our approach become apparent when the impact on that step is assessed.

Check node operations in decoding NB-LDPC codes consist of two major sub-steps: permutation and convolution of probability vectors \cite{qspa-fft:declerq}. In the permutation sub-step, the simpler one of the two, symbol probability vectors received by the check node are permuted, where the permutations are defined by the respective edge weights. Since $\Theta^Q_{p^{r-m}}$ is constructed using the smallest $\alpha$-connected set $\Theta_{p^{r-m}}$, it is clear from Lemma 5 that CPVs will also have to be permuted similarly. Thus, complexity of the permutation step will not be significantly affected in the proposed approach.

In order to understand how our approach changes the convolution sub-step, consider the simple case of a degree 3 check node in a code over $\mathbb{F}_{p^{r}}$, where the parity-check equation is $v_1+v_2+v_3=0$. A convolution has to be carried out using the incoming symbol probability vectors of $v_1$ and $v_2$ for computing the outgoing symbol probability vector of $v_3$, $\bf{p}^s_{v_3}$. Since these vectors are of length $r$, convolution will be of complexity order $\mathcal{O}(p^{2r})$. Now assume we have to compute some $i$'th CPV of $v_3$, $\bf{p}^c_{v_3,i}$. This computation also only requires the incoming $i$'th CPVs of the remaining two variable nodes. Note that these are of length $m$, where $m~|~r$. As all quotient groups in $\Theta^Q_{p^{r-m}}$ are isomorphic to the additive group of $\mathbb{F}_{p^m}$, computation of $\bf{p}^c_{v_3,i}$ should be the same as the convolution sub-step at a check node of a code over $\mathbb{F}_{p^m}$. Thus, complexity is now only of order $\mathcal{O}(p^{2m})$. However, with $|\Theta^Q_{p^{r-m}}|=\frac{p^r-1}{p^m-1}$, that many CPVs will have to be computed, resulting in an overall complexity of $p^{2m}\times \frac{p^r-1}{p^m-1} \approx \mathcal{O}(p^{m+r})$. Nevertheless, particularly for the cases where $m\ll r$, this is a significant reduction of complexity.

Motivated by the observation that using coset probability vectors instead of symbol probability vectors can allow faster decoding of NB-LDPC codes, we will provide a more detailed analysis of these complexity advantages in Section VI. In the following subsection, we explain how to \textit{expand} a graph over $\mathbb{F}_{p^r}$ into one over $\mathbb{F}_{p^m}$ so that CPVs can be used in decoding. 

\subsection{Graph Expansion}
We assume that a Tanner graph of a code over $\mathbb{F}_{p^r}$ is to be expanded into a graph over $\mathbb{F}_{p^m}$, where $m~|~r$. The set of quotient groups, $\Theta^Q_{p^{r-m}}$, will be of cardinality $\frac{p^r-1}{p^m-1}$. Each $Q_i\in \Theta^Q_{p^{r-m}}$ is isomorphic to $\{\mathbb{F}_{p^m},+\}$, and in decoding, an associated CPV has to be used. Observations on how CPVs impact decoding suggest that it is possible to simply replace each node in the original graph, i.e., the so-called $\mathbb{F}_{p^r}$ nodes, with $\frac{p^r-1}{p^m-1}$ $\mathbb{F}_{p^m}$ nodes. Each variable node over $\mathbb{F}_{p^m}$ would represent some CPV, and check nodes would calculate their estimates. How the set of $\mathbb{F}_{p^m}$ variable and check nodes of a single neighboring variable-check node pair of the original graph are connected will depend on the original edge weight, as evident from Lemma 5.

Consider a check node and a variable node in the original graph, connected with an edge of weight $\alpha^k\in \mathbb{F}_{p^r}$. According to Lemma 5, $Q_i \in \Theta^Q_{p^{r-m}}$ becomes a permutation of some $Q_j\in \Theta^Q_{p^{r-m}}$ when multiplied with $\alpha^k$. Then, in the expansion, the $\mathbb{F}_{p^m}$ variable node representing the $i$'th CPV should be connected to the $\mathbb{F}_{p^m}$ check node calculating estimates of the $j$'th CPV. As $Q_i$ turns into a \textit{permutation} of $Q_j$, CPVs transmitted along this edge is permuted as well. Thus, this is a  2-step process, where first the set of CPVs are permuted, and then each CPV is permuted within itself. From the point-of-view of expansion, it is equivalent to connecting the set of $\mathbb{F}_{p^m}$ variable nodes with the set of check nodes using edges labeled with elements from $\mathbb{F}_{p^m}$.

As an example, consider parity-check equation $\rho$ from a code over $\mathbb{F}_{2^4}$, where $\alpha$ denotes a primitive of the field.
\[ \rho \Rightarrow \alpha^4 v_1 + \alpha v_2 = 0 \tag{3} \]
Fig. 1 presents the initial expansion for $\rho$. The shaded graph is the original Tanner graph over $\mathbb{F}_{2^4}$, and the graph beneath is the expansion over $\mathbb{F}_{2^2}$. In both graphs, circles denote variable nodes and squares denote check nodes. Note that $\omega$ is a primitive of $\mathbb{F}_{2^2}$ and that edges in the expanded graph are labeled with $\mathbb{F}_{2^2}$ elements.

\begin{figure}[ht]
\centering
\includegraphics[scale=0.15]{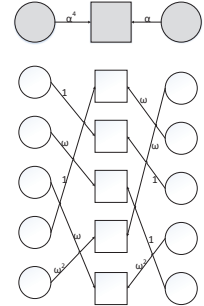}
\captionof{figure}{Initial Expansion}
\end{figure}

\vspace{-4mm}
Since each $Q_i\in \Theta^Q_{p^{r-m}}$ contains different groupings of the same set of symbols, the associated CPV contains some information about all other CPVs. Unfortunately, initial graph expansion is unable to capture these dependencies. In order to clearly visualize the relationships between CPVs, we propose an \textit{alternate representation} of $\mathbb{F}_{p^r}$ symbols below.

As each $Q_i\in \Theta^Q_{p^{r-m}}$ is isomorphic to $H=\{ \mathbb{F}_{p^m},+\}$, every coset in $Q_i$ maps to some element of $H$. We define the \textit{value} of some $\gamma \in \mathbb{F}_{p^{r}}$ \textit{with respect to} some $Q_i\in \Theta^Q_{p^{r-m}}$ as the element of $H$ that maps to the coset containing $\gamma$. Since $|\Theta^Q_{p^{r-m}}|=\frac{p^r-1}{p^m-1}$, using values defined with respect to each $Q_i$, $\gamma$ can be uniquely represented as a vector of $\frac{p^r-1}{p^m-1}$ elements of $H$. For example, consider the case of $\mathbb{F}_{2^4}$ and $\mathbb{F}_{2^2}$. Table I presents the $\frac{2^4-1}{2^2-1}=5$ quotient groups in $\Theta^Q_{2^{4-2}}$, where each coset is listed under the $H_{2^2}=\{\mathbb{F}_{2^2},+\}$ element we map to it in the isomorphism between its quotient group and $H_{2^2}$. Alternative representations of $\mathbb{F}_{2^4}$ elements as length $5$ vectors over $\mathbb{F}_{2^2}$ are listed in Table II. Note that a position $i$ in these vectors map to quotient group $Q_i$ as given in Table I.


\begin{table}[ht]
    \centering
    \setlength\tabcolsep{4pt}
    \begin{tabular}{|c|c|c|c|c|}
        \hline
        & $0$ & $1$ & $\omega$ & $\omega^2$ \\
        \hline
        $Q_0$ & $0,\alpha,\alpha^6,\alpha^{11}$ & $1,\alpha^4,\alpha^{12},\alpha^{13}$ & $\alpha^2,\alpha^3,\alpha^5,\alpha^9$ & $\alpha^7,\alpha^8,\alpha^{10},\alpha^{14}$ \\
        $Q_1$ & $0,1,\alpha^5,\alpha^{10}$ & $\alpha,\alpha^2,\alpha^4,\alpha^8$ & $\alpha^6,\alpha^7,\alpha^9,\alpha^{13}$ & $\alpha^3,\alpha^{11},\alpha^{12},\alpha^{14}$ \\
        $Q_2$ & $0,\alpha^4,\alpha^9,\alpha^{14}$ & $1,\alpha,\alpha^3,\alpha^7$ & $\alpha^5,\alpha^6,\alpha^8,\alpha^{12}$ & $\alpha^2,\alpha^{10},\alpha^{11},\alpha^{13}$\\
        $Q_3$ & $0,\alpha^2,\alpha^7,\alpha^{12}$ & $1,\alpha^8,\alpha^9,\alpha^{11}$ & $\alpha,\alpha^5,\alpha^{13},\alpha^{14}$ & $\alpha^3,\alpha^4,\alpha^6,\alpha^{10}$\\
        $Q_4$ & $0,\alpha^3,\alpha^8,\alpha^{13}$ & $1,\alpha^2,\alpha^6,\alpha^{14}$ & $\alpha^4,\alpha^5,\alpha^7,\alpha^{11}$ & $\alpha,\alpha^9,\alpha^{10},\alpha^{12}$\\
         \hline
    \end{tabular}
    \vspace{2mm}
    \caption{Quotient Groups in $\Theta^Q_{2^{4-2}}$}
\end{table}

\vspace{-12pt}

\begin{table}[ht]
    \centering
    \begin{tabular}{|c|c|c|c|}
         \hline
         $0 : 0 0 0 0 0$ & $\alpha^3 : \omega \omega^2 1 \omega^2 0$ & $\alpha^7 : \omega^2 \omega 1 0 \omega$ & $\alpha^{11} : 0 \omega^2 \omega^2 1 \omega$\\
         $1 : 1 0 1 1 1$ & $\alpha^4 : 1 1 0 \omega^2 \omega$ & $\alpha^8 : \omega^2 1 \omega 1 0$ & $\alpha^{12} : 1 \omega^2 \omega 0 \omega^2$ \\
         $\alpha : 0 1 1 \omega \omega^2$ & $\alpha^5 : \omega 0 \omega \omega \omega$ & $\alpha^9 : \omega \omega 0 1 \omega^2$ & $\alpha^{13} : 1 \omega \omega^2 \omega 0$\\
         $\alpha^2 : \omega 1 \omega^2 0 1$ &  $\alpha^6: 0 \omega \omega \omega^2 1$ & $\alpha^{10} : \omega^2 0 \omega^2 \omega^2 \omega^2$ & $\alpha^{14} : \omega^2 \omega^2 0 \omega 1$ \\
         \hline
    \end{tabular}
    \vspace{2mm}
    \caption{Alternate Representations of Symbols in $\mathbb{F}_{2^4}$}
\end{table}

\vspace{-4mm}

Note that the sixteen vectors in Table II form a 2-dimensional space over $\mathbb{F}_{2^2}$. In channel coding terms, they are the 16 codewords of a $(2,5)$ linear code over $\mathbb{F}_{2^2}$. Thus, values of some $\gamma \in \mathbb{F}_{2^4}$ with respect to 2 $Q_i$'s in $\Theta^Q_{2^{4-2}}$ are sufficient to derive the remaining three. The dependancies could easily be captured through the parity-check equations of the code.

In the general case of $\mathbb{F}_{p^r}$ and $\mathbb{F}_{p^m}$, alternate representation vectors would form a $\frac{r}{m}$ dimensional vector space, or in other words a $(\frac{p^r-1}{p^m-1},\frac{r}{m})$ code, over $\mathbb{F}_{p^m}$. The $\frac{p^r-1}{p^m-1}$ $\mathbb{F}_{p^m}$ nodes of every $\mathbb{F}_{p^r}$ variable node would form this code, and since each such instance only involves the set of $\mathbb{F}_{p^m}$ nodes of a single $\mathbb{F}_{p^r}$ variable node, we refer to it as the \textit{`local code'}. We propose using the parity-check matrix (PCM) of the local code, $^p\mathbb{H}^{r,m}_L$, to succinctly represent the dependancies between CPVs. Note that codes with parameters of the form $(\frac{p^r-1}{p^m-1},\frac{r}{m})$ are from the family of non-binary simplex codes. Since dual of such a code is the $(\frac{p^r-1}{p^m-1},\frac{p^r-1}{p^m-1}-\frac{r}{m})$ Hamming code over $\mathbb{F}_{p^m}$ \cite{book:lin}, parity-check equations would be Hamming codewords, and the PCM would contain $\frac{p^r-1}{p^m-1}-\frac{r}{m}$ of those. As an example, the `local' PCM for the case of $\mathbb{F}_{2^4}$ and $\mathbb{F}_{2^2}$, $^2\mathbb{H}_{L}^{4,2}$, which consists of 3 codewords of the $(3,5)$ Hamming code over $\mathbb{F}_{2^2}$, is given below.
{{
\[ ^2\mathbb{H}_{L}^{4,2} = \begin{bmatrix}
    1\hspace{0.1cm} 1\hspace{0.1cm} 1\hspace{0.1cm} 0\hspace{0.1cm} 0\\
    1\hspace{0.1cm} \omega\hspace{0.1cm} 0\hspace{0.1cm} 1\hspace{0.1cm} 0\\
    1\hspace{0.1cm} \omega^2\hspace{0.1cm} 0\hspace{0.1cm} 0\hspace{0.1cm} 1\\
\end{bmatrix} \tag{4}\] }}

From the perspective of the expansion, parity-check equations due to the local code are a set of additional check nodes, which have to be added to the expanded graph. Since the set of $\mathbb{F}_{p^m}$ nodes of every $\mathbb{F}_{p^r}$ variable node forms one instance of the local code, $\frac{p^r-1}{p^m-1}-\frac{r}{m}$ additional check nodes representing the local PCM have to be added \textit{per variable node of the original graph}. Adding such a large number of check nodes might seem to increase complexity, but it should be noted that these new nodes are of low degrees. Since the dual is a Hamming code, it should always be possible to find a set of degree 3 parity-check equations for the local PCM. With each new check node only being connected with a subset of the $\frac{p^r-1}{p^m-1}$ nodes of one $\mathbb{F}_{p^r}$ variable node, they will be referred to as `local check nodes' here onward. Check nodes resulting from expanding $\mathbb{F}_{p^r}$ nodes will be called `regular check nodes'.

Performance with any form of iterative message passing decoding is dependant on features of the graph used, and it is well-known that short cycles in the graph negatively impact decoding. While with binary codes, where the edge labels are all $1$, effects of cycles depend only on their length, edge labels themselves have an impact in the non-binary case \cite{absorbing-sets:amiri}. Particularly troublesome there are the cycles created by sub-matrices in the PCM that are not of full-rank \cite{absorbing-sets:amiri}-\cite{cycle-counting:cho}. Decoding performance of the expanded graph may be improved if the subgraph induced by the local PCM is free of these undesirable structures as much as possible. Although one could use a canonical generator matrix of a Hamming code as the local PCM, the graph induced may not entirely suit iterative decoding. In such a scenario, raw operations can be carried out on the matrix until a `better' one is obtained. This is particularly important when the expansion results in a \textit{binary} graph ($p=2$ and $m=1$), since then \textit{any} short cycle is detrimental for decoding. We have explored this case separately in \cite{cpbd}, where the expansion is arrived at in a slightly different way than the more general approach presented here. In the non-binary case ($m>1$), it might not be possible to remove all short cycles, and one might have to be satisfied with a local PCM only free of short cycles not satisfying the `full-rank condition', such as $^2\mathbb{H}^{4,2}_L$ given in (4). Decoding scheme we propose in the next section employs a technique to further reduce the possible negative effects of cycles among local check nodes.


Local check nodes enable us to adequately capture the various dependancies between CPVs, and adding those to the expanded graph wraps up the expansion. Different steps necessary for expanding a graph over $\mathbb{F}_{p^r}$ into one over $\mathbb{F}_{p^m}$, where $m~|~r$, can be summarized as follows.
\begin{enumerate}
    \item {
        Obtain the smallest set of $\alpha$-connected subgroups $\Theta_{p^{r-m}}$, using the homomorphism presented in definition 3, and following the steps outlined in lemma 4. Use that to derive the set of quotient groups $\Theta^Q_{p^{r-m}}$.
    }
    \item {
        Map cosets of each $Q_i \in \Theta^Q_{p^{r-m}}$ with elements of $H=\{ \mathbb{F}_{p^m},+ \}$. Use these isomorphisms to obtain the alternate representation vectors of $\mathbb{F}_{p^r}$ elements.
    }
    \item {
        Find a PCM more suited to iterative decoding for the code formed by alternate representation vectors. 
    }
    \item {
        Expand each node in the original graph into $\frac{p^r-1}{p^m-1}$ $\mathbb{F}_{p^m}$ nodes. Connect the new variable and check nodes and label the edges, based on edge labels in the original graph.
    }
    \item {
        Add local check nodes to represent the local PCM found in step 3.
    }
\end{enumerate}

Fig. 2 presents the complete expansion for the earlier example of parity-check equation $\rho$, given by (3). $\mathbb{F}_{2^4}$ graph is shaded grey, and the expansion to $\mathbb{F}_{2^2}$ is depicted in white. Circles represent variable nodes, squares regular check nodes, and hexagons local check nodes. Note that in the interest of a clearer figure, edge labels are only shown for the instances where they are $\neq1$.

\begin{figure}[ht]
\centering
\includegraphics[scale=0.1]{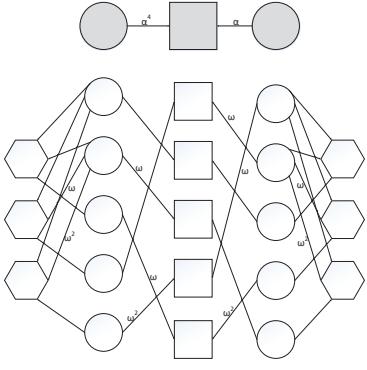}
\captionof{figure}{Final Expansion}
\end{figure}

\vspace{-6mm}

\section{Decoding Scheme}

An iterative message passing decoding algorithm that utilizes the Tanner graph representation of NB-LDPC codes can be used with the expanded graph. Advantage herein is the expansion being over a smaller field than the original graph, leading to a lower decoding complexity. Note that a few different options are available for expanding a graph over $\mathbb{F}_{p^r}$, one for each factor of $r$. Each of these would offer a different complexity-performance trade-off, which may suit different applications. 

Any generic decoding algorithm can be applied straightforwardly to decode the expanded graph with some simple modifications. In the following, we present these modifications, and explain why they are required. Note that the explanation is from the perspective of a soft decision decoding (SDD) algorithm, such as QSPA \cite{nb-ldpc:davey-mckay}, and its many variations \cite{qspa-fft:declerq}-\cite{trellis-ems:li}, but can be also applied to other algorithms such as majority-logic decoding \cite{mlgd:lin}. \\

1) \textit{Computing Channel Estimates}: Any SDD algorithm has to be initialized with probability estimates based on channel observations. In QSPA and its variants, for initializing the decoder, variable nodes compute channel estimates that are of the form of symbol probability vectors. In the proposed expansion, each variable node represents some CPV. Thus, when using SDD algorithms on expanded graphs, it is required to compute initial estimates for CPVs.

Note that each coset contains a subset of elements in $\mathbb{F}_{p^r}$. This makes computing initial estimates of CPVs quite straight-forward, i.e., probability of a ($\mathbb{F}_{p^r}$) variable node belonging to a particular coset of some subgroup can be calculated by simply summing up probabilities of those symbols that belong to the coset. Equation (5) presents this computation, where $\bf{p}^s_n$ is the symbol probability vector of original variable node $n$, $\bf{p}^c_{n,i}$ is the $i$'th CPV of that node, $C^j_i$ is the $j$'th coset in $i$'th quotient group, and $a^j_k$'s are $\mathbb{F}_{p^r}$ elements in that coset.
\[ \bf{p}^c_{n,i}(j)=\sum_{a^j_k\in C^j_i} \bf{p}^s_n(a^j_k) \hspace{1cm} j=0,1,..,(p^m-1) \tag{5} \]

In most practical applications, decoders operate on either $\log$ or $\log$-likelihood ratio (LLR) domain, due to hardware stability concerns \cite{hw-decoders:spagnol}. In such a case, $\bf{p}^c_{n,i}$ has to be converted to the desired domain, for example 
\[ \underline{L}_{n,i}^c(j)=\log \frac{\bf{p}^c_{n,i}(j)}{\bf{p}^c_{n,i}(0)} \hspace{1cm} j=0,1,..,(p^m-1) \tag{6}\]

For one $\mathbb{F}_{p^r}$ variable node, $\frac{p^r-1}{p^m-1}$ $\mathbb{F}_{p^m}$ nodes that represent CPVs have to be initialized as in (5). Only a single symbol probability vector, corresponding to the single $\mathbb{F}_{p^r}$ element transmitted through the channel, will be used for all those computations. This implies that channel observation is only sufficient to initialize $\frac{r}{m}$ $\mathbb{F}_{p^m}$ symbols. However, in this approach, there are $\frac{p^r-1}{p^m-1}$ nodes that are initialized. Thus, channel observations are duplicated and dependencies are created between initial estimates of CPVs. Any error in channel estimates gets multiplied, and propagates through the graph, leading to performance losses.

Recall the fact that the set of $\mathbb{F}_{p^m}$ nodes of a single $\mathbb{F}_{p^r}$ variable node are `connected' via the local code. Local code is an $(\frac{p^r-1}{p^m-1},\frac{r}{m})$ code, and therefore, $\frac{r}{m}$ out of the $\frac{p^r-1}{p^m-1}$ $\mathbb{F}_{p^m}$ nodes can be thought of as representing \textit{information symbols}, and others \textit{parity symbols}. We propose first picking a suitable set of $\frac{r}{m}$ nodes to represent information symbols, and initializing only these as in (5) and (6). For rest of the nodes, those that represent parity symbols of the local code, an additional scaling factor $\delta$ ($0\leq \delta \leq 1$) will be used in (6). Our simulation results show that this modification helps in reducing propagation of errors in channel information, but $\delta$ has to be optimized per code. Equation (7) presents this modification. 
\[ \underline{L}_{n,i}^c(j)=\delta . \log \frac{\bf{p}^c_{n,i}(j)}{\bf{p}^c_{n,i}(0)} \hspace{1cm} j=0,1,..,(p^m-1) \tag{7} \]

After initialization, operations of the decoder would be similar to those of a decoder for a code over $\mathbb{F}_{p^m}$ except for a couple of minor modifications that are explained in the following. \\

2) \textit{Distinguishing Local Checks from Regular Checks:} 
Expanded graphs contain two different types of check nodes; local check nodes that represent dependencies between CPVs, and regular ones, resulting from expanding check nodes in the original graph. Here, local check nodes are only connected with $\mathbb{F}_{p^m}$ nodes of a single $\mathbb{F}_{p^r}$ variable node, whereas a regular check node will only be connected with one such. Thus, local check nodes do not represent relationships between different variables of the original code, and regular ones represent only those. This means that estimates from the two types of check nodes are based on two separate linear codes, and treating them similarly may not be the best approach to take.

As discussed in Section III $B$, local PCM may contain some short cycles, and these would be present in the expanded graph among the local check nodes. Estimates computed by a check node involved in one such cycle in two different iterations will be correlated with each other to some degree. This can make the estimates `over-confident' of a variable node taking a particular value.

Taking into consideration the need to distinguish between estimates of local and regular check nodes, and also since local check nodes could be involved in short cycles, we propose using another scaling factor $\psi$ (0$<$$\psi$$<$1) with estimates of local check nodes.  In the literature, similar approaches have been taken to mitigate effects of short cycles with satisfactory results, for example in \cite{rs:jiang}.

Combining probability estimates with this modification, at some variable node $i$ of the expanded graph during $k$'th decoding iteration, is given by (8). There, $\underline{L}_i$ is the initial estimate for node $i$, $\underline{R}^{(k)}_i$ is the combined estimate, and $\underline{r}^{(k)}_{j\xrightarrow{}i}$ is the estimate sent from $j$'th check node to $i$'th variable node, in $k$'th iteration. $\underline{L}_i, \underline{R}^{(k)}_i$ and  $\underline{r}^{(k)}_{j\xrightarrow{}i}$ are all length $p^m$ vectors of $\log$ or LLR values. $N^r_i$ and $N^l_i$ are, respectively, sets of regular and local check nodes in the neighborhood of node $i$.
\[ \underline{R}^{(k)}_i = \underline{L}_i + \sum_{j\in N^r_i} \underline{r}^{(k)}_{j\xrightarrow{}i} + \psi . \sum_{j\in N^l_i} \underline{r}^{(k)}_{j\xrightarrow{}i} \tag{8} \]

Similar to scaling factor $\delta$ used in initialization, $\psi$ also has to be optimized per code.\\

3) \textit{Testing for Convergence:} In iterative decoding of NB-LDPC codes, a tentative decision is taken by every variable node in each iteration to test whether the decoder has converged to a valid codeword. If so, then the check-sum at every check node should be zero, and the decoding process can be terminated. Same approach may be taken when decoding on expanded graphs. Tentative decision at each variable node would be the $\mathbb{F}_{p^m}$ element most likely for the node, and check-sums would be computed at all check nodes, including local ones. Output of the decoder would be a vector of $\mathbb{F}_{p^m}$ elements that's $\frac{p^r-1}{p^m-1}$ times longer than the original code length. Original codeword can be recovered by mapping each set of $\frac{p^r-1}{p^m-1}$ $\mathbb{F}_{p^m}$ elements to a single $\mathbb{F}_{p^r}$ element, via the `local' code, as discussed in Section III $B$.

Even though we replace each $\mathbb{F}_{p^r}$ node with $\frac{p^r-1}{p^m-1}$ $\mathbb{F}_{p^m}$ nodes, just $\frac{r}{m}$ $\mathbb{F}_{p^m}$ elements are sufficient to represent a single $\mathbb{F}_{p^r}$ element, which is also evident from the local code. This observation leads to a slightly easier approach to checking convergence. Rather than deciding on all $\mathbb{F}_{p^m}$ nodes of a single $\mathbb{F}_{p^r}$ variable node, we propose only using the $\frac{r}{m}$ nodes selected as the `information symbols' of the local code. Most likely $\mathbb{F}_{p^m}$ elements of these would map to a single $\mathbb{F}_{p^r}$ element, once more through the local code. Check-sums of original parity-check equations can then be computed with these $\mathbb{F}_{p^r}$ elements. Note that even though now check-sums are computed over the larger field, computations involve only simple field arithmetic, and also there will be a significant reduction in the number of computations required when compared with the straight-forward approach. \\

With these three modifications, any iterative soft-decoding algorithm \cite{nb-ldpc:davey-mckay}-\cite{trellis-ems:li} proposed for NB-LDPC codes may be used with expanded graphs. This allows a large number of decoding strategies. For applications where decoding latency is the primary concern, a simplification of QSPA, such as min-max decoding \cite{minmax:savin}, can be used with an expanded graph, thereby achieving the complexity gains of both the simplification and the expansion. Section V presents some results from simulations where a few of these different strategies were evaluated.


\section{Simulation Results}
In this section, we compare error-correcting performance of decoding schemes discussed in Section IV against some existing decoding algorithms for NB-LDPC codes. We consider different expansions of the same Tanner graph (different $m$ for a fixed graph), and use QSPA \cite{nb-ldpc:davey-mckay}, and one of its well-known simplifications, min-max decoding \cite{minmax:savin}, with each expansion. QSPA and min-max decoding are also used on the original graph, along with max-log-SP algorithm \cite{qspa-llr:wymeersch}, which is a special case of the extended min-sum (EMS) algorithm \cite{ems:declerq-fossorier}, where $n_m$ and $n_c$ are set to the maximum possible values of the size of the field and check-node degree, respectively. All algorithms were implemented in LLR domain \cite{qspa-llr:wymeersch}, and simulations were done over the BI-AWGN channel, with maximum decoding iterations of 50 for all. Algorithms over expanded graphs were used with the modifications proposed in Section IV, and scaling factors $\delta$ and $\psi$ were optimized through simulations. In the following, we use the algorithm along with the field size to refer to different decoding setups, for example, we let $\mathbb{F}_{p^r}$-QSPA denote QSPA on a graph over $\mathbb{F}_{p^r}$, and etc.

\begin{figure}[ht]
\centering
\includegraphics[scale=0.5]{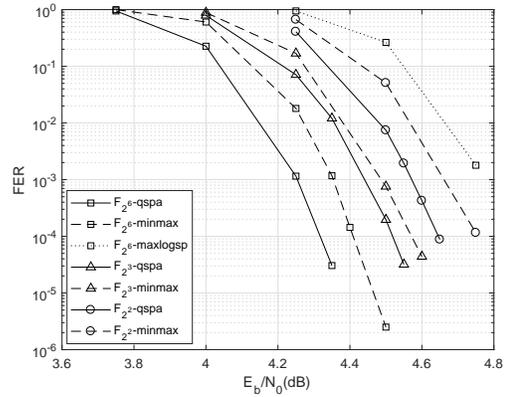}
\captionof{figure}{FER Perf. with a (1998,1776) code over $GF(2^6)$ (\( \mathcal{C}_1 \))}
\end{figure}

Fig.~3 shows FER performance of decoding schemes with \( \mathcal{C}_1 \), a rate $0.89$ code over $\mathcal{F}_{2^6}$, of 1998 symbols in length. Code was generated through random re-labeling of a regular binary LDPC code of column weight 4, obtained from \cite{database:mckay}. 

In Fig.~3, we observe that decoding algorithms over expanded graphs perform close to the best known decoder, QSPA over the original graph. In fact, QSPA over the $\mathbb{F}_{2^3}$ expansion performs within $0.2$dB of $\mathbb{F}_{2^6}$-QSPA, at a FER of $10^{-4}$. When using the $\mathbb{F}_{2^2}$ expansion, this widens slightly to $0.3$dB. While min-max decoding over the original graph has a gap of only about $0.08$dB with $\mathbb{F}_{2^6}$-QSPA, it should be noted that decoding is still over $\mathbb{F}_{2^6}$, and thus, it is more complex than QSPA over expanded graphs, as made evident in Section VI. Interestingly, other simplification of QSPA, max-log-SP algorithm, is outperformed by all proposed decoding schemes, although it operates in the original field. Max-log-SP shows a gap of about $0.55$dB with $\mathbb{F}_{2^6}$-QSPA, at a FER of $10^{-3}$. We also evaluate performance of min-max decoding over expanded graphs, which is quite satisfactory. In the case of $\mathbb{F}_{2^3}$ expansion, min-max only has a gap of $0.06$dB with $\mathbb{F}_{2^3}$-QSPA, while the gap between $\mathbb{F}_{2^2}$-QSPA and $\mathbb{F}_{2^2}$-min-max is around $0.1$dB. Interestingly, these two decoding setups, which have complexity advantages of expansion and simplification, manage to outperform the max-log-SP algorithm over the original graph. Optimum values for scaling factors $(\delta,\psi)$ were found to be $(0.75,0.25)$ for $\mathbb{F}_{2^3}$-QSPA and $\mathbb{F}_{2^2}$-QSPA, $(0,0.3)$ for $\mathbb{F}_{2^3}$-min-max, and $(0,0.4)$ for $\mathbb{F}_{2^2}$-min-max.

\begin{figure}[ht]
\centering
\includegraphics[scale=0.5]{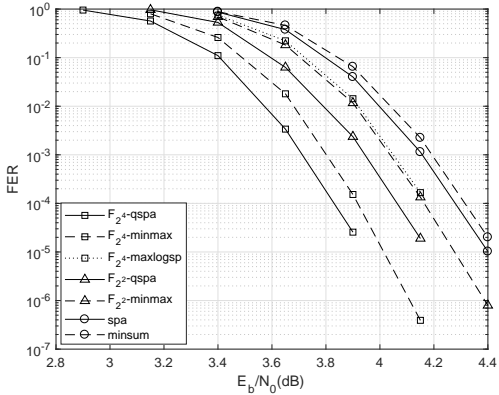}
\captionof{figure}{FER Perf. with a (1000,861) code over $GF(2^4)$ (\( \mathcal{C}_2 \))}
\end{figure}

Fig. 4 illustrates the FER performance of proposed schemes with a rate 0.861 code over $\mathbb{F}_{2^4}$, of 1000 symbols in length (\( \mathcal{C}_2 \)). The $\mathcal{C}_2$ was generated by re-labeling a regular binary graph of column weight 3, constructed with the progressive edge growth algorithm \cite{peg:hu}. For this code, we consider expansions over $\mathbb{F}_{2^2}$ and $\mathbb{F}_{2}$. Expansion over $\mathbb{F}_{2}$ is of special interest, since it results in a \textit{binary} graph. When using this binary graph, we replace QSPA and min-max decoding with SPA and its well-known simplification, min-sum algorithm (MSA). Unique features and advantages offered by the binary expansion have been explored separately in \cite{cpbd}.

Fig. 4 shows that the performance losses of the proposed schemes are quite small in this case as well. Gap between using QSPA on the original graph and its $\mathbb{F}_{2^2}$ expansion is less than $0.3$dB at a FER of $10^{-4}$. Loss of replacing QSPA by its simplification min-max decoding is about $0.1$dB for both original and expanded graphs. With $\mathcal{C}_2$, max-log-SP algorithm seems to perform a bit better than with $\mathcal{C}_1$. Here, its performance is very similar to that of using min-max algorithm on $\mathbb{F}_{2^2}$ expansion, with a gap of close to $0.4$dB with $\mathbb{F}_{2^4}$-QSPA, at a FER of $10^{-4}$. When compared with QSPA on the original graph, using SPA on the binary graph results in a $0.5$dB loss in performance. Simplifying SPA to MSA only loses a further $0.05$dB. Although a $0.5$dB loss seems significant, as explored in \cite{cpbd}, decoding on a binary expansion provides unique advantages in decoding complexity and hardware implementations. Optimum values for scaling factors $(\delta,\psi)$ here were $(0.5,0.25)$ for $\mathbb{F}_{2^2}$-QSPA, SPA, and MSA, and $(0,0.3)$ for $\mathbb{F}_{2^2}$-min-max.

Simulation results show that decoding algorithms implemented on proposed graph expansions are capable of performing quite close to those using the original graph. For any algorithm, performance gap of decoding on the expanded graph and using the original widens when the size of the field used for the expansion decreases. With a few different graph expansions possible, many decoding options become available for any given code. As discussed in the next Section, all these decoding schemes provide attractive complexity gains, with different levels of performance-complexity trade-offs. 


\section{Decoding Complexity}

In the following, we analyze the complexities of some decoding schemes on expanded graphs. We consider implementing the two popular versions of QSPA, LLR-QSPA \cite{qspa-llr:wymeersch} and FFT-QSPA \cite{qspa-fft:declerq}, and also min-max decoding \cite{minmax:savin} on proposed expansions and compare them in terms of complexity with the same algorithms implemented on the original graph. Since NB-LDPC codes are most often defined over finite fields of characteristic 2 \cite{nb-ldpc:davey-mckay}, a code over $\mathbb{F}_{2^r}$, where $r$ has a factor $m$, is used in the complexity analysis. Complexities of the two major steps in iterative decoding, check node operations and variable node operations, are compared separately. For the comparison, we consider operations at a single node of each type during one iteration. Since the proposed expansions replace each node over $\mathbb{F}_{2^r}$ with $E_f = \frac{2^r-1}{2^m-1}$ nodes, complexity of all those is the total complexity for the decoding schemes on expanded graphs. As explained in Section III $B$, these graphs also have the additional feature of \textit{local} check nodes. Since $L_f = \frac{2^r-1}{2^m-1}-\frac{r}{m}$ such nodes are included per variable node of the original graph, their complexities are included with that of variable nodes.

At hardware level, apart from the number of operations, the type of operation also affects the complexity. It is well-known that operations such as multiplications are more complex than comparisons \cite{hw-decoders:spagnol}. Therefore, we consider the number of operations of a few different types; comparisons (Comp), additions/subtractions (Add), multiplications/divisions (Mult) and table look-ups (LUT). Note that $\max^*$ operation in LLR-QSPA can be performed with one comparison, two additions, and one table look-up \cite{qspa-llr:wymeersch}, and that transformation between $\log$ and probability domain, required in FFT-QSPA, can be carried out with look-up tables. It has also been assumed that the forward-backward approach \cite{minmax:savin} is used in check node operations of the three algorithms. Further, cost of permuting probability vectors has been disregarded, since its impact on total complexity is negligible.

Table III lists complexities of check node operations in each decoding setup, while Table IV considers variable node operations. Average degrees of a check node and a variable node in the original graph are denoted with $d_c$ and $d_v$, while $d_{l}$ denotes the average degree of a local check node. As discussed in Section III $B$, local PCM is formed with Hamming codewords, and therefore it should always be possible to set $d_{l}=3$. Due to these new check nodes, average variable node degree would slightly increase in the expanded graphs, and we denote this new value with $\widetilde{{d}_v}$, given by
\[ \widetilde{{d}_v} = d_v + \frac{L_f\times d_l}{E_f} \tag{9}\]
Substituting the values for $E_f, L_f$ and $d_v$ yields
\[ \widetilde{{d}_v} = d_v + 3 - 3\times\frac{r(2^m-1)}{m(2^r-1)} \tag{10}\]
Note that degrees of \textit{regular} check nodes in the expanded graphs remain $d_c$. When presenting complexities of decoding schemes on these graphs, we let $E_f, L_f$ and $\widetilde{{d}_v}$ denote the number of new nodes per original node, number of local check nodes, and average variable node degree, respectively.

\begin{table}[ht]
\renewcommand{\arraystretch}{1}
\centering
\captionof{table}{Check Node Complexity}
\begin{tabular}{ |c|c|c|c|c| } 
\hline
 \multirow{2}{*}{\textit{Algorithm}} & \multicolumn{4}{c|}{\textit{Check Node Operations}}\\
 \cline{2-5}
    & {Comp} & {Add} & {Mult} & {LUT}\\
\hline
\hline
$\mathbb{F}_{2^r}$ & $(3d_c-4)\times$ & $(3d_c-4)\times$ & $0$ & $(3d_c-4)\times$\\
-LLR-QSPA & $2^r(2^r-1)$ & $2^r(3.2^r-2)$ & & $2^r(2^r-1)$\\
\hline
$\mathbb{F}_{2^m}$ & $E_f(3d_c-4)\times$ & $E_f(3d_c-4)\times$ & $0$ & $E_f(3d_c-4)\times$\\
-LLR-QSPA & $2^m(2^m-1)$ & $2^m(3.2^m-2)$ & & $2^m(2^m-1)$\\
\hline
\hline
$\mathbb{F}_{2^r}$ & $0$ & $2d_c\times$ & $(2d_c-1)\times$ & $2d_c\times$\\
-FFT-QSPA & & $2^rr$ & $2^r$ & $2^r$\\
\hline
$\mathbb{F}_{2^m}$ & $0$ & $E_f.2d_c\times$ & $E_f(2d_c-1)\times$ & $E_f.2d_c\times$\\
-FFT-QSPA & & $2^mm$ & $2^m$ & $2^m$\\
\hline
\hline
$\mathbb{F}_{2^r}$ & $(3d_c-4)\times$ & 0 & 0 & 0\\
-Min-Max & $2^r(2.2^{r}-1)$ & & & \\
\hline 
$\mathbb{F}_{2^m}$ & $E_f(3d_c-4)\times$ & 0 & 0 & 0\\
-Min-Max & $2^m(2.2^{m}-1)$ & & & \\
\hline
\end{tabular}
\end{table}

From Table III, it can be seen that complexity gains of proposed schemes at check node operations depend on the decoding algorithm being used. For both LLR-QSPA and min-max decoding, using an expanded graph instead of the original results in a significant reduction in complexity, while for FFT-QSPA, the gains are modest. In the case of LLR-QSPA, using the original graph requires approximately $3d_c\times 2^{2r}$ comparisons, additions, and table look-ups, which results in an overall complexity of $\mathcal{O}(2^{2r})$. However, with the expansion over $\mathbb{F}_{2^m}$, there are only approximately $3d_c\times 2^{r+m}$ operations of each type, which reduces overall complexity to $\mathcal{O}(2^{r+m})$. This is a significant gain, especially in the cases with a large $r$, and we feel that, as a trade-off, the small performance losses observed in Section V are justifiable. Using an expanded graph can reduce the complexity order from $\mathcal{O}(2^{2r})$ to $\mathcal{O}(2^{r+m})$ in check node operations of min-max decoding as well. It should be noted that although they are of the same complexity order, min-max decoding is simpler than LLR-QSPA, since only comparisons are required. Gains of the proposed scheme reduce in the case of FFT-QSPA. Here, the number of multiplications and table look-ups required are almost the same (approximately $2d_c\times 2^r$) when using the original graph or an expanded one. There is a slight reduction in the number of additions though, from approximately $2d_c\times 2^rr$ to $2d_c\times 2^rm$. Thus, the overall complexity of FFT-QSPA on an expanded graph is $\mathcal{O}(2^rm)$, slightly lower than $\mathcal{O}(2^rr)$ on the original graph.


\begin{table}[ht]
\renewcommand{\arraystretch}{1}
\centering
\captionof{table}{Variable Node Complexity}
\begin{tabular}{ |c|c|c|c|c| } 
\hline
 \multirow{2}{*}{\textit{Algorithm}} & \multicolumn{4}{c|}{\textit{Variable Node Operations}}\\
 \cline{2-5}
    & {Comp} & {Add} & {Mult} & {LUT}\\
\hline
\hline
$\mathbb{F}_{2^r}$ & $2^r-1$ & $2d_v\times$ & $0$ & $0$\\
-LLR-QSPA & & $2^r$ & & \\
\hline
$\mathbb{F}_{2^m}$ & $(r/m)\times$ & $E_f.2\widetilde{d_v} \times$ & $0$ & $0$ \\
-LLR-QSPA & $(2^m-1)$ & $2^m$ & & \\
\hdashline
 & $5L_f\times$ & $5L_f\times$ & $0$ & $5L_f\times$ \\
Local Checks & $2^m(2^m-1)$ & $2^m(3.2^m-2)$ & & $2^m(2^m-1)$\\
\hline
\hline
$\mathbb{F}_{2^r}$ & $2^r-1$ & $2d_v\times$ & $0$ & $0$\\
-FFT-QSPA & & $2^r$ & & \\
\hline
$\mathbb{F}_{2^m}$ & $(r/m)\times$ & $E_f.2\widetilde{d_v} \times$ & $0$ & $0$ \\
-FFT-QSPA & $(2^m-1)$ & $2^m$ & & \\
\hdashline
 & $0$ & $6L_f\times$ & $5L_f\times$ & $6L_f\times$\\
Local Checks & & $2^mm$ & $2^m$ & $2^m$ \\
\hline
\hline
$\mathbb{F}_{2^r}$ & $(d_v+1)\times$ & $3d_v\times$ & 0 & 0\\
-Min-Max & $2^r$ & $2^r$ & & \\
\hline 
$\mathbb{F}_{2^m}$ & $(E_f.\widetilde{d_v}+r/m)$ & $E_f.3\widetilde{d_v}\times $ & $0$ & $0$\\
-Min-Max & $\times 2^m$ & $2^m$ & & \\
\hdashline
 & $5L_f\times$ & $0$ & $0$ & $0$\\
Local Checks & $2^m(2.2^m-1)$ & & & \\
\hline
\end{tabular}
\end{table}

When considering variable node operations of decoding schemes on expanded graphs, we include the complexity of the $L_f$ local check nodes added for each \textit{original} variable node. Note that complexity of one such node can be derived by substituting $d_l=3$ as the node degree, and $2^m$ as the field size, in the expressions for the respective algorithm in Table III. Due to this additional cost, complexity at variable nodes are \textit{higher} in proposed schemes. However, this complexity increase is not sufficiently high to completely offset the gain obtained at check node operations, especially for LLR-QSPA and min-max decoding. As Table IV shows, complexity orders of these algorithms change from $\mathcal{O}(2^r)$ on the original graph to $\mathcal{O}(2^{r+m})$ on an expanded one, while in Table III, this change is from $\mathcal{O}(2^{2r})$ to $\mathcal{O}(2^{r+m})$ at check node operations. Hence, the overall gain is still significant for LLR-QSPA and min-max decoding, especially for larger values of $r$. In the case of FFT-QSPA, the complexity increase is comparatively smaller, from $\mathcal{O}(2^r)$ to $\mathcal{O}(2^rm)$. Since its gain at check nodes was also quite modest, the overall complexity gain would be minimal.

Tables III and IV demonstrate that decoding on expanded graphs is advantageous in terms of asymptotic complexity, while the actual performance gains would depend on parameters of the code used, such as field sizes, code length, rate, and average node degrees. In Table V, we consider complexities of some decoding schemes used in Section V with $\mathcal{C}_1$, a code over $\mathbb{F}_{2^6}$ with the codeword length $1998$ and code rate $0.89$. In this case, the original graph is over $\mathbb{F}_{2^6}$, and expansions over $\mathbb{F}_{2^3}$ and $\mathbb{F}_{2^2}$ are used for decoding. Table V presents complexities of using LLR-QSPA, FFT-QSPA, and min-max decoding on all three graphs, in terms of number of operations of each type per iteration. For decoding schemes over expansions, we also present the number of operations required as a percentage of the requirement when using the same algorithm with the original graph.

\begin{table}[ht]
\renewcommand{\arraystretch}{1}
\centering
\captionof{table}{Number of Operations per Iteration with $\mathcal{C}_1$}
\begin{tabular}{ |c|c|c|c|c| } 
\hline
 \multirow{2}{*}{\textit{Algorithm}} & \multicolumn{4}{c|}{\textit{Number of Operations} ($\times 10^5$)}\\
 \cline{2-5}
    & {Comp} & {Add} & {Mult} & {LUT}\\
\hline
\hline
$\mathbb{F}_{2^6}$-LLR-QSPA & 932.17 & 2817.73 & - & 930.91 \\
\hline
$\mathbb{F}_{2^3}$-LLR-QSPA & 155.8 & 507.01 & - & 155.52 \\
 & ($\approx 17\%$) & ($\approx 18\%$) &  & ($\approx 17\%$) \\
\hline
$\mathbb{F}_{2^2}$-LLR-QSPA & 79.94 & 287.93 & - & 79.76 \\
 & ($\approx 8\%$) & ($\approx 10\%$) &  & ($\approx 8\%$) \\
\hline
\hline
$\mathbb{F}_{2^6}$-FFT-QSPA & 1.26 & 71.61 & 10.09 & 10.23\\
\hline
$\mathbb{F}_{2^3}$-FFT-QSPA & 0.28 & 72.89 & 16.94 & 18.22 \\
 & ($\approx 22\%$) & ($\approx 101\%$) & ($\approx 168\%$) & ($\approx 178\%$)\\
\hline
$\mathbb{F}_{2^2}$-FFT-QSPA & 0.18 & 66.17 & 20.43 & 22.06 \\
 & ($\approx 14\%$) & ($\approx 92\%$) & ($\approx 202\%$) & ($\approx 215\%$)\\
\hline
\hline
$\mathbb{F}_{2^6}$-Min-Max & 1882.99 & 15.35 & - & - \\
\hline 
$\mathbb{F}_{2^3}$-Min-Max & 342.7 & 27.33 & - & - \\
 & ($\approx 18\%$) & ($\approx 178\%$) &  & \\
\hline
$\mathbb{F}_{2^2}$-Min-Max & 197.38 & 33.09 & - & - \\
 & ($\approx 10\%$) & ($\approx 215\%$) &  & \\
\hline
\end{tabular}
\end{table}

In Table V, we observe that using LLR-QSPA on expanded graphs offers exceptional complexity gains for $\mathcal{C}_1$. Less than $20\%$ of the operations for the original graph are required when using the $\mathbb{F}_{2^3}$ expansion. This reduces further with the $\mathbb{F}_{2^2}$ expansion, to less than $10\%$. These gains correspond to speed-ups of \textit{more than 5 times} in the $\mathbb{F}_{2^3}$ case, and \textit{more than 10 times} in the $\mathbb{F}_{2^2}$ case. Considering that the performance losses, as shown in Section V, are only $0.2$dB and $0.3$dB, the complexity gains are very attractive. With FFT-QSPA though, using expansions are not particularly advantageous. Only gain of $\mathbb{F}_{2^3}$ expansion, when compared with using the algorithm on the original $\mathbb{F}_{2^6}$ graph, is in the number of comparisons required. Both decoding setups use a similar number of additions, while the setup on the expanded graph needs significantly more multiplications and table look-ups. This is due to the operations of local check nodes, which are absent in the original graph. With $\mathbb{F}_{2^2}$ expansion, the number of comparisons reduces further, and the number of additions used is also slightly lesser than that of the $\mathbb{F}_{2^6}$ case. Since $\mathbb{F}_{2^2}$ expansion has more local check nodes than the $\mathbb{F}_{2^3}$ one, the number of multiplications and table look-ups have increased significantly. Thus, for $\mathcal{C}_1$, using FFT-QSPA with any of the two expansions is \textit{more complex} than implementing on the original graph. The case of min-max decoding is very similar to that of LLR-QSPA; complexity gains are significant, and they are higher when the size of the field used is smaller. Due to local check node operations, the number of additions in proposed schemes is higher than in the original algorithm. Nevertheless, since the reduction in the number of comparisons is much higher in magnitude, min-max decoding on expanded graphs is significantly less complex.

Majority of existing algorithms are of complexity order $\mathcal{O}(2^{2r})$ for a code over $\mathbb{F}_{2^r}$, and implementing those algorithms on graph expansions results in significant complexity gains with minimal performance losses. For algorithms whose complexity order is not polynomial in field size, such as FFT-QSPA, the new strategy may not be advantageous. But as \cite{hw-decoders:spagnol} pointed out, out of the two variants of QSPA, LLR-QSPA is more suitable for hardware implementations, due to better numerical stability of LLR domain operations. Therefore, the strategy proposed in this paper could be applied to reduce decoding complexity in most practical applications that adopt NB-LDPC codes. In particular, our proposed strategy enables to decode a code defined over a large field using a graph over a much smaller field, while providing a good performance and complexity tradeoff, leading to a practical solution to decoding NB-LDPC codes.

\section{Conclusions}

In this paper, we proposed a new method to expand a Tanner graph of a NB-LDPC code over $\mathbb{F}_{p^r}$ into a graph over $\mathbb{F}_{p^m}$, where $m$ is a factor of $r$. Most decoding algorithms proposed for NB-LDPC codes can be adapted to use these expanded graphs with simple modifications. This offers a number of different decoding options for any given code, with a different performance-complexity trade-off. Simulation results show that, in general, decoding on expanded graphs provide significant complexity gains, while performance losses are minimal. It may be interesting to note that the proposed  expansion could be useful in other applications beyond decoding NB-LDPC codes.


%



\ifCLASSOPTIONcaptionsoff
  \newpage
\fi


\begin{thebibliography}{1}

\bibitem{ldpc:gallager}
R. G. Gallager, ``Low-density parity-check codes", \textit{IRE Transactions on Information Theory}, vol. IT-8, pp. 21-28, Jan. 1962
\bibitem{ldpc:mckay}
D. J. C. MacKay, ``Good error-correcting codes based on very sparse matrices", \textit{IEEE Transactions on Information Theory}, vol. 46, no. 2, pp. 399-431, Mar. 1999
\bibitem{nb-ldpc:davey-mckay}
M. C. Davey, and D. J. C. Mackay, ``Low-density parity check codes over $GF(q)$", \textit{IEEE Communication Letters}, vol. 2, no. 6, pp. 165-167, June 1998
\bibitem{qspa-fft:declerq}
L. Barnault, and D. Declerq, ``Fast decoding algorithm for LDPC over $GF(2^q)$", \textit{Proceedings of IEEE Information Theory Workshop}, Paris, France, Apr. 2003
\bibitem{qspa-llr:wymeersch}
H. Wymeersch, H. Steendam, and M. Moeneclaey, ``Log-domain decoding of LDPC codes over $GF(q)$",  \textit{Proceedings of IEEE International Conference on Communications}, Paris, France, Jun. 2004
\bibitem{ems:declerq-fossorier}
D. Declercq, and M. Fossorier, ``Decoding algorithms for nonbinary LDPC codes over $GF(q)$", \textit{IEEE Transactions on Communications}, vol. 55, no. 4, pp. 633-643, Apr. 2007
\bibitem{minmax:savin}
V. Savin, ``Min-Max decoding for non binary LDPC codes", \textit{Proceedings of IEEE International Symposium on Information Theory}, Toronto, Canada, July 2008
\bibitem{trellis-ems:li}
E. Li, D. Declercq, and K. Gunnam, ``Trellis-based extended min-sum algorithm for non-binary LDPC codes and its hardware structure", \textit{IEEE Transactions on Communications}, vol. 61, no. 7, pp. 2600-2611, July 2013
\bibitem{trellis-minmax:lacruz}
J. O. Lacruz, F. Garcia-Herrero, D. Declercq, and J. Valls, ``Simplified trellis min–max decoder architecture for nonbinary low-density parity-check codes", \textit{IEEE Transactions on Very Large Scale Integration (VLSI) Systems}, vol. 23, no. 9, pp. 1783-1792, Sep. 2015
\bibitem{bec:savin}
V. Savin, ``Binary linear-time erasure decoding for non-binary LDPC codes", \textit{Proceedings of IEEE Information Theory Workshop}, Taormina, Italy, Oct. 2009
\bibitem{generalized-rep:Yu}
Y. Yu, W. Chen, J. Li, X. Ma, and B. Bai, ``Generalized binary representation for the nonbinary LDPC code with decoder design", \textit{IEEE Transactions on Communications}, vol. 62, no. 9, pp. 3070-3083, Sep. 2014
\bibitem{bit-level:zhang}
M. Zhang, K. Cai, Q. Huang, and S. Yuan, ``On bit-level decoding of nonbinary LDPC codes", \textit{IEEE Transactions on Communications}, vol. 66, no. 9, pp. 3736-3748, Sep. 2018
\bibitem{book:rothman}
J. J. Rothman, ``Advanced modern algebra", 1st ed. Prentice Hall, 2003, pp. 116-218
\bibitem{book:lin}
S. Lin, and D. J. Costello, ``Error control coding", Upper Saddle
River, NJ, USA: Pearson Education, 2004
\bibitem{absorbing-sets:amiri}
B. Amiri, J. Kliewer, and L. Dolecek, ``Analysis and enumeration of absorbing sets for non-binary graph-based codes", \textit{IEEE Trans. on Comm.}, vol. 62, no. 2, pp. 398-409, Feb. 2014
\bibitem{cycle-counting:cho}
S. Cho, K. Cheun, and K. Yang, ``A message-passing algorithm for counting short cycles in nonbinary LDPC codes", \textit{Proc. of IEEE ISIT}, Vail, CO, USA, June 2018
\bibitem{cpbd} 
V. B. Wijekoon, Emanuele Viterbo, Yi Hong, R. Micheloni, and A. Marelli, ``A Novel Graph Expansion and a Decoding Algorithm for NB-LDPC Codes", \textit{IEEE Trans. on Comm.}, vol. 68, no. 3, pp. 1358 - 1369, Mar. 2020
\bibitem{mlgd:lin}
Chao-Yu Chen, Qin Huang, Chi-chao Chao, and Shu Lin, ''Two low-complexity reliability-based message-passing algorithms for decoding non-binary LDPC codes", \textit{IEEE Transactions on Communications}, vol. 58, no. 11, Nov. 2010
\bibitem{hw-decoders:spagnol}
C. Spagnol, E.M. Popovici, and W.P. Marnane , ``Hardware implementation of $GF(2^m)$ LDPC decoders", \textit{IEEE Trans. Circuits Syst. I}, vol. 56, no. 12, pp. 2609-2620, Mar. 2009
\bibitem{rs:jiang}
J. Jiang, and K. R. Narayanan, ``Iterative soft-input soft-output decoding of Reed-Solomon codes by adapting the parity-check matrix", \textit{IEEE Trans. on Inf. Th.}, vol. 52, no. 8, pp. 3746-3756, Aug. 2006
\bibitem{database:mckay}
D. J. C. Mackay, ``Encyclopedia of Sparse Graph Codes",  [Online]. Available: http://www.inference.org.uk/mackay/codes/data.html.
\bibitem{peg:hu}
X.-Y. Hu, E. Eleftheriou, and D.M. Arnold, ``Regular and irregular progressive edge-growth tanner graphs", \textit{IEEE Trans. on Inf. Th.}, vol. 51, no. 1, pp. 386-398, Jan. 2005

\end{thebibliography}
\end{document}